\documentclass{article}

\usepackage{arxiv}

\usepackage{graphicx}
\usepackage[T1]{fontenc}
\usepackage[utf8]{inputenc}
\usepackage[english]{babel}
\usepackage{hyperref}  % professional-quality tables
\usepackage{amsmath}
\usepackage{amsthm}
\usepackage{amsfonts}       % blackboard math symbols
\usepackage[ruled,vlined]{algorithm2e}
\SetKwInOut{Input}{input}
\SetKw{KwRet}{return}
\SetKw{Break}{break}
\usepackage{xparse}
\usepackage{todonotes}

\graphicspath{{./figures/}}

\newcommand{\suchthat}{\;\ifnum\currentgrouptype=16 \middle\fi|\;}

\NewDocumentCommand{\abs}{som}{%
  \IfNoValueTF{#2}{
    \IfBooleanTF{#1}{%
      \left\lvert\mskip0.3\thinmuskip#3\mskip0.3\thinmuskip\right\rvert%
    }{%
      \lvert\mskip0.3\thinmuskip#3\mskip0.3\thinmuskip\rvert%
    }%
  }{%
    \mathopen{#2\lvert}\mskip0.3\thinmuskip#3\mskip0.3\thinmuskip\mathclose{#2\rvert}
  }%
}%

\newtheorem{theorem}{Theorem}

\newtheorem{conjecture}[theorem]{Conjecture}

\newtheorem{lemma}[theorem]{Lemma}
\newtheorem{corollary}[theorem]{Corollary}

\newcommand\txteq[1]{\stackrel{\text{#1}}{=}}
\newcommand\txtle[1]{\stackrel{\text{#1}}{\le}}

\DeclareMathOperator*{\argmax}{arg\,max}

\title{Counting problems in trees, with applications to fixed points of
  cellular automata}

\author{%
  \href{https://orcid.org/0000-0001-9964-8816}{%
    \includegraphics[height=0.8em]{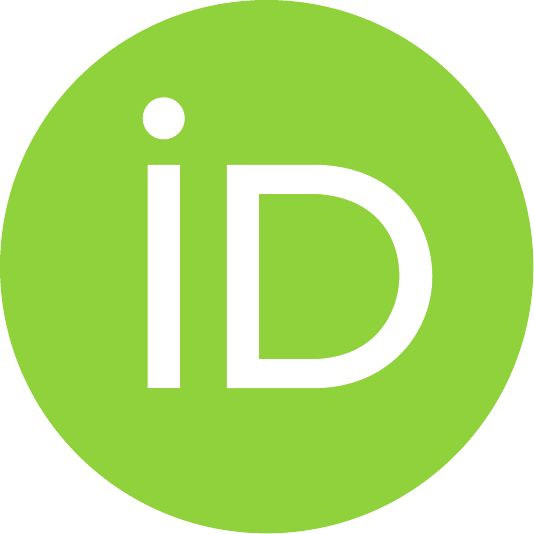}%
    \hspace{1mm}%
    Volker Turau%
  } \\
  Institute of Telematics \\
  Hamburg University of Technology \\
  21073 Hamburg, Germany \\
  \texttt{turau@tuhh.de} }

\begin{document}

\maketitle

\begin{abstract}
  Cellular automata are synchronous discrete dynamical systems used to
  describe complex dynamic behaviors. The dynamic is based on local
  interactions between the components, these are defined by a finite
  graph with an initial node coloring with two colors. In each step,
  all nodes change their current color synchronously to the least/most
  frequent color in their neighborhood and in case of a tie, keep
  their current color. After a finite number of rounds these systems
  either reach a fixed point or enter a 2-cycle. The problem of
  counting the number of fixed points for cellular automata is
  \#P-complete. In this paper we consider cellular automata defined by
  a tree. We propose an algorithm with run-time $O(n\Delta)$ to count
  the number of fixed points, here $\Delta$ is the maximal degree of
  the tree. We also prove upper and lower bounds for the number of
  fixed points. Furthermore, we obtain corresponding results for pure
  cycles, i.e., instances where each node changes its color in every
  round. We provide examples demonstrating that the bounds are sharp.
  The results are proved for the minority and the majority model.
  \keywords{Tree cellular automata\and Fixed points \and Counting
    problems.}
\end{abstract}
%
%

% They are also often used as a model to simulate a broad
% variety of physical, social, and sociotechnical distributed
% infrastructures.

\section{Introduction}
\label{sec:intro}

A widely used abstraction of classical distributed systems such as
multi-agent systems are {\em graph automata}. They evolve over time
according to some simple local behaviors of its components. They
belong to the class of synchronous discrete-time dynamical systems. A
common model is as follows: Let $G$ be a graph, where each node is
initially either black or white. In discrete-time rounds, all nodes
simultaneously update their color based on a predefined local rule.
Locality means that the color associated with a node in round $t$ is
determined by the colors of the neighboring nodes in round $t-1$. As a
local rule we consider the minority and the majority rule that arises
in various applications and as such have received wide attention in
recent years, in particular within the context of information
spreading. Such systems are also known as graph cellular automata. It
is well-known \cite{Goles:1980,Poljak:1983} that they always converge
to configurations that correspond to cycles either of length 1 --
a.k.a.\ fixed points -- or of length 2, i.e., such systems eventually
reach a stable configuration or toggle between two configurations.

% , e.g., how many nodes must initially be black so that eventually
% all nodes are black

One branch of research so far uses the assumption that the initial
configuration is random. Questions of interest are on the expected
stabilization time of this process \cite{Zehmakan:2019} and the
dominance problem \cite{Peleg:2002}. Fogelman et al.\ proved that the
stabilization time is $O(n^2)$ \cite{Fogelman:1983}.

% Frischknecht et al.\ showed that this bound is tight, up to some
% poly-logarithmic factor \cite{Frischknecht:2013}.

In this paper we focus on counting problems related to cellular
automata, in particular counting the number of fixed points and pure
2-cycles, i.e., instances where each node changes its color in every
round. This research is motivated by applications of so-called Boolean
networks (BN) \cite{Kauffman:1993}, i.e., discrete-time dynamical
systems, where each node (e.g., gene) takes either 0 (passive) or 1
(active) and the states of nodes change synchronously according to
regulation rules given as Boolean functions. Since the problem of
counting the fixed points of a BN is in general 
\#P-complete \cite{Akutsu:1998,Agha_05,Bridoux_22}, it is interesting
to find graph classes, for which the number of fixed points can be
efficiently determined. These counting problems have attracted a lot of
research in recent years \cite{Aracena2014,Aledo2022,Mezzini2020}.

We consider tree cellular automata, i.e., the defining graphs are
finite trees. The results are based on a characterization of fixed
points and pure 2-cycles for tree cellular automata \cite{Turau_2022}.
The authors of \cite{Turau_2022} describe algorithms to enumerate all
fixed points and all pure cycles. Since the number of fixed points and
pure 2-cycles can grow exponentially with the tree size, these
algorithms are unsuitable to efficiently compute these numbers. We
prove the following theorem.

\begin{theorem}\label{theo:1}
  The number of fixed points and the number of pure 2-cycles of a tree
  with $n$ nodes and maximal node degree $\Delta$ can be computed in
  time $O(n\Delta)$.
\end{theorem}

We also prove the following theorem with upper and lower bounds for
the number of fixed points of a tree improving results of
\cite{Turau_2022} (parameter $r$ is explained in
Sec.~\ref{sec:lowerFix}). In the following, the $i^{th}$ Fibonacci
number is denoted by $\mathbb{F}_i$.

\begin{theorem}\label{theo:2}
  A tree with $n$ nodes, diameter $D$ and maximal node degree $\Delta$
  has at least $\max \left( 2^{r/2+1}, 2\mathbb{F}_D\right)$ and at most
  $\min\left( 2^{n-\Delta}, 2\mathbb{F}_{n-\lceil \Delta/2\rceil}\right)$ fixed
  points.
\end{theorem}

For the number of pure cycles we prove the following result, which
considerably improves the bound of \cite{Turau_2022}.

\begin{theorem}\label{theo:3}
  A tree with maximal degree $\Delta$ has at most
  $\min\left( 2^{n-\Delta}, 2\mathbb{F}_{\lfloor n/2\rfloor}\right)$ pure
  2-cycles.
\end{theorem}

We provide examples demonstrating ranges where these bounds are sharp.
All results hold for the minority and the majority rule. We also
formulate several conjectures about counting problems and propose
future research directions.

% Finally, we look at general configurations with period $2$. We show
% that for each configuration $c$ of this type each tree decomposes into
% subtrees, such that $c$ induces either a fixed point or a pure
% configuration on each subtree. The subtrees allow to define a hyper
% structure of a tree, called the {\em block tree}. We identify a subset
% $E_{block}(T)$ of $\mathcal{P}(E)$ and show that the elements of
% $E_{block}(T)$ correspond one-to-one with the block trees of $T$.
% Since a tree can have several pure colorings, a block tree does not
% uniquely define a coloring. We define a subclass of 2-cycles called
% canonical colorings and prove that there is a direct correspondence
% between $E_{block}(T)$ and canonical colorings. The characterization
% of $E_{block}(T)$ is not as simple as in the above cases, since
% $E_{block}(T)$ does not have the hereditary property.

% All results are obtained for the minority and the majority model. We
% first present the results for the minority model. In
% Section~\ref{sec:majority} we discuss the necessary deviations for the
% majority model.

\section{State of the Art}
\label{sec:state-art}

The analysis of fixed points of minority/majority rule cellular
automata received limited attention so far. Kr{\'a}lovi{\v{c}}
determined the number of fixed points of a complete binary tree for
the majority process \cite{Rastislav:2001}. For the majority rule he
showed that this number asymptotically tends to $4n(2\alpha)^n$, where
$n$ is the number of nodes and $\alpha \approx 0.7685$. Agur et al.\
did the same for ring topologies \cite{Agur:1988}, the number of fixed
point is in the order of $\Phi^n$, where $\Phi = (1 +\sqrt{5})/2$. In
both cases the number of fixed points is an exponentially small
fraction of all configurations.

A related concept are Boolean networks (BN). They have been
extensively used as mathematical models of genetic regulatory
networks. The number of fixed points of a BN is a key feature of its
dynamical behavior. A gene is modeled by binary values, indicating two
transcriptional states, active or inactive. Each network node operates
by the same nonlinear majority rule, i.e., majority processes are a
particular type of BN \cite{Veliz:2012}. The number of fixed points is
an important feature of the dynamical behavior of a BN
\cite{Aracena:2008}. It is a measure for the general memory storage
capacity. A high number implies that a system can store a large amount
of information, or, in biological terms, has a large phenotypic
repertoire \cite{Agur:1991}. However, the problem of counting the
fixed points of a BN is in general \#P-complete \cite{Akutsu:1998}.
There are only a few theoretical results to efficiently determine this
set \cite{Irons:2006}. Aracena determined the maximum number of fixed
points regulatory Boolean networks, a particular class of BN
\cite{Aracena:2008}.

Recently, Nakar and Ron studied the dynamics of a class of synchronous
one-dimensional cellular automata for the majority rule
\cite{Nakar:2022a}. They proved that fixed points and 2-cycles have a
particular spatially periodic structure and give a characterization of
this structure. Concepts related to fixed points of the
minority/majority process have been analyzed. A partition
$(S,\bar{S})$ of the nodes of a graph is called a global defensive
$0$-alliance or a monopoly if $\abs{N_S(v)}\ge \abs{N_{\bar{S}}(v)}$
for each node $v$ \cite{Mishra:2006}. Thus, a fixed point of the
minority/majority process induces a $0$-alliance (but not conversely).

% A related concept is that of $2$-community structure. Bazgan et al.\
% prove that each tree has a connected $2$-community structure that can
% be found in linear time \cite{Bazgan:2010}.

Most research on discrete-time dynamical systems on graphs is focused
on bounding the stabilization time. Good overviews for the majority
(resp.\ minority) process can be found in \cite{Zehmakan:2019} (resp.\
\cite{Papp:2019}). Rouquier et al.\ studied the minority process in the
asynchronous model, i.e., not all nodes update their color
concurrently \cite{Rouquier:2011}. They showed that the stabilization
time strongly depends on the topology and observe that the case of
trees is non-trivial.

% For $F\subseteq E$ let $\mathcal{C}_T(F)$ be the set of connected
% components of $T\!\setminus\! F$. We define a tree $\mathcal{T}_F$
% with nodes $\mathcal{C}_T(F)$ and edges $F$. An edge of $(u,w)\in F$
% connects components $T_1,T_2 \in \mathcal{C}_T(F)$ if and only if
% $u\in T_1$ and $w\in T_2$.

% For a set $S$ we denote by $\mathcal{P}(S)$ the power set of $S$,
% i.e., the set of all subsets of $S$.

\subsection{Notation}
Let $T=(V,E)$ be a finite, undirected tree with $n=\abs{V}$. The
maximum degree of $T$ is denoted by $\Delta(T)$, the diameter by
$D(T)$. The parameter $T$ is omitted in case no ambiguity arises. A
{\em star graph} is a tree with $n-1$ leaves. A {\em $l$-generalized
  star graph} is obtained from a star graph by inserting $l-1$ nodes
into each edge, i.e., $n=l\Delta+1$. For $F\subseteq E$ and $v\in V$
denote by $deg_{F}(v)$ the number of edges in $F$ incident to $v$.
Note that $deg_{F}(v)\le deg(v)$. For $i \ge 2$ denote by $E^i(T)$ the
set of edges of $T$, where each end node has degree at least $i$. For
$v\in V$ denote the set of $v$'s neighbors by $N(v)$. For
$e=(v_1,v_2)\in E^2(T)$ let $T_i$ be the subtree of $T$
consisting of $e$ and the connected component of $T\setminus e$ that
contains $v_i$. We call $T_i$ the {\em constituents} of $T$ for $e$.
$T_1$ and $T_2$ together have $n+2$ nodes. We denote the $i^{th}$
Fibonacci number by $\mathbb{F}_i$, i.e.,
$\mathbb{F}_0=0, \mathbb{F}_1=1$, and
$\mathbb{F}_i=\mathbb{F}_{i-1}+\mathbb{F}_{i-2}$.

\section{Synchronous Discrete-Time Dynamical Systems}
Let $G=(V,E)$ be a finite, undirected graph. A coloring $c$ assigns to
each node of $G$ a value in $\{0,1\}$ with no further constraints on
$c$. Denote by $\mathcal{C}(G)$ the set of all colorings of $G$, i.e.,
$\abs{\mathcal{C}(G)}=2^{\abs{V}}$. A transition process $\mathcal{M}$
is a mapping
$\mathcal{M}: \mathcal{C}(G) \longrightarrow \mathcal{C}(G)$. Given an
initial coloring $c$, a transition process produces a sequence of
colorings $c, \mathcal{M}(c), \mathcal{M}(\mathcal{M}(c)),\ldots$. We
consider two transition processes: {\em Minority} and {\em Majority}
and denote the corresponding mappings by $\mathcal{MIN}$ and
$\mathcal{MAJ}$. They are local mappings in the sense that the new
color of a node is based on the current colors of its neighbors. To
determine $\mathcal{M}(c)$ the local mapping is executed in every {\em
  round} concurrently by all nodes. In the minority (resp.\ majority)
process each node adopts the minority (resp.\ majority) color among
all neighbors. In case of a tie the color remains unchanged (see
Fig.~\ref{fig:minority}). Formally,
the minority process is defined for a node $v$ as follows:
\[\mathcal{MIN}(c)(v) = \begin{cases}
    c(v) & \text{if } \abs{N^{c(v)}(v)} \le  \abs{N^{1-c(v)}(v)}\\
    1-c(v)  & \text{if } \abs{N^{c(v)}(v)} > \abs{N^{1-c(v)}(v)} 
  \end{cases}\] $N^i(v)$ denotes the set of $v$'s neighbors with color
$i$ ($i=0,1$). The definition of $\mathcal{MAJ}$ is similar, only the
binary operators $\le$ and $>$ are reversed. Some results hold for
both the minority and the majority process. To simplify notation we
use the symbol $\mathcal{{M}}$ as a placeholder for $\mathcal{{MIN}}$
and $\mathcal{{MAJ}}$. 
\begin{figure}[h]
  \hfill
  \includegraphics[scale=0.7]{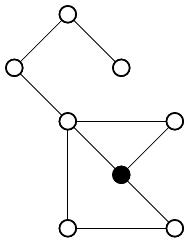}
  \hfill
  \includegraphics[scale=0.7]{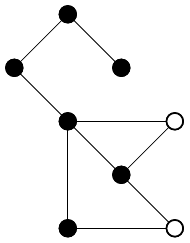}
  \hfill
  \includegraphics[scale=0.7]{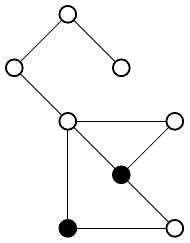}
  \hfill
  \includegraphics[scale=0.7]{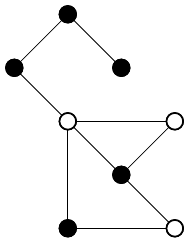}
  \hfill
  \includegraphics[scale=0.7]{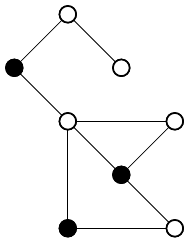}
  \hfill
  \includegraphics[scale=0.7]{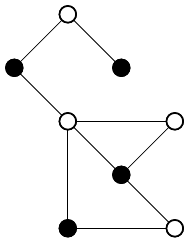}
  \hfill\null
  \caption{For the coloring on the left $\mathcal{{MIN}}$ reaches
    after 5 rounds a fixed point. $\mathcal{{MAJ}}$ reaches for the
    same initial coloring after one round a monochromatic
    coloring.}\label{fig:minority}
\end{figure}

Let $c\in \mathcal{C}(G)$. If $\mathcal{M}(c)=c$ then $c$ is called a
{\em fixed point}. It is called a {\em 2-cycle} if
$\mathcal{M}(c)\not=c$ and $\mathcal{M}(\mathcal{M}(c))=c$. A 2-cycle
is called {\em pure} if $\mathcal{M}(c)(v) \not =c(v)$ for each node
$v$ of $G$, see Fig.~\ref{fig:example_def}. Denote by
$\mathcal{F}_{\mathcal{M}}(G)$ (resp.\ $\mathcal{P}_{\mathcal{M}}(G)$)
the set of all $c\in \mathcal{C}(G)$ that constitute a fixed point
(resp.\ a pure 2-cycle) for $\mathcal{M}$. 

% We define $\mathcal{F}_{\mathcal{M}}(G)^+$
% (resp.\ $\mathcal{P}_{\mathcal{M}}(G)^+$) as the subsets of
% $\mathcal{F}_{\mathcal{M}}(G)$ (resp.\ $\mathcal{P}_{\mathcal{M}}(G)$)
% which assign to a globally distinguished node $v^\ast$ color $0$.
% Hence, if $c\in \mathcal{F}_{\mathcal{M}}(G)$ then either $c$ or the
% complement of $c$ is in $\mathcal{F}_{\mathcal{M}}(G)^+$.

% $c$ is called {\em monochromatic} if all nodes have the
% same color, i.e., $c(v)=c(w)$ for all $v,w\in V$ (see
% Fig.~\ref{fig:example_def}). $c$ is called {\em independent} if the
% color of each node is different from the colors of all its neighbors.
% Clearly, a monochromatic (resp.\ independent) coloring is a fixed
% point for the majority (resp.\ minority) process. An edge $(v,w)$ is
% called {\em monochromatic} for $c$ if $c(v)=c(w)$ otherwise it is
% called {\em multi-chromatic}.

\begin{figure}[h]
  \hfill
  \includegraphics[scale=0.8]{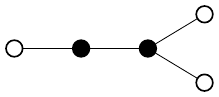}
  \hfill
  \includegraphics[scale=0.8]{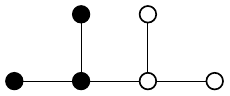}% Wird ersetzt
   \hfill%\hspace*{10mm}
  \includegraphics[scale=0.8]{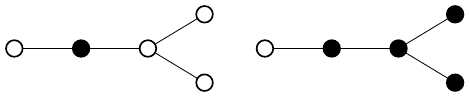}
  \hfill\null
  \caption{Examples for the $\mathcal{MIN}$ rule. The coloring of the
    first (resp.\ second) tree is a fixed point (resp.\ a pure
    2-cycle). The right two colorings are a non-pure
    2-cycle.}\label{fig:example_def}
\end{figure}

Let $T$ be a tree. The following results are based on a
characterization of $\mathcal{F}_{\mathcal{M}}(T)$ and
$\mathcal{P}_{\mathcal{M}}(T)$ by means of subsets of $E(T)$
\cite{Turau_2022}. Let $E_{fix}(T)$ be the set of all {\em
  $\mathcal{F}$-legal} subsets of $E(T)$, where $F\subseteq E(T)$ is
{\em $\mathcal{F}$-legal} if $2deg_{F}(v) \le deg(v)$ for each
$v\in V$. Each $\mathcal{F}$-legal set is contained in $E^2(T)$, hence
$\abs{E_{fix}(T)}\le 2^{\abs{E^2(T)}}$. Theorem 1 of \cite{Turau_2022}
proves that $\abs{\mathcal{F}_{\mathcal{M}}(T)} = 2\abs{E_{fix}(T)}$,
see Fig.~\ref{fig:example_fix_col}. Let $E_{pure}(T)$ be the set of
all {\em $\mathcal{P}$-legal} subsets of $E(T)$, where
$F\subseteq E(T)$ is {\em $\mathcal{P}$-legal} if
$2deg_{F}(v) < deg(v)$ for each $v\in V$. Thus, {\em
  $\mathcal{P}$-legal} subsets are contained in $E^3(T)$ and therefore
$\abs{E_{pure}(T)}\le 2^{\abs{E^3(T)}}$. Theorem 4 of
\cite{Turau_2022} proves that
$\abs{\mathcal{P}_{\mathcal{M}}(T)} = 2\abs{E_{pure}(T)}$. For the
tree in Fig.~\ref{fig:example_fix_col} we have
$E_{pure}(T)=\{\emptyset\}$, thus
$\abs{\mathcal{P}_{\mathcal{M}}(T)}=2$. The pure colorings are the two
monochromatic colorings. Given these results it is unnecessary to
treat $\mathcal{{MIN}}$ and $\mathcal{{MAJ}}$ separately. To determine
the number of fixed points (resp.\ pure 2-cycles) it suffices to
compute $\abs{E_{fix}(T)}$ (resp.\ $\abs{E_{pure}(T)}$).
\begin{figure}[h]
  \begin{center}
  \includegraphics[scale=0.9]{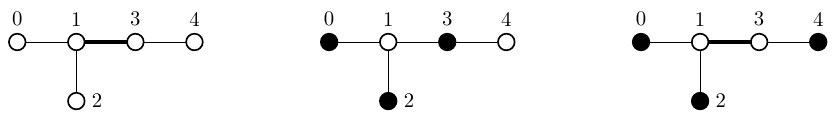}    
  \end{center}
  \caption{A tree $T$ with $E_{fix}(T)=\{\emptyset, \{(1,3)\}\}$ and
    the corresponding fixed points for $\mathcal{MIN}$, the other two
    can be obtained by inverting colors.}\label{fig:example_fix_col}
\end{figure}

\section{Fixed Points}
In this section we propose an efficient algorithm to determine
$\abs{\mathcal{F}_{\mathcal{M}}(T)}$, we provide upper and lower
bounds for $\abs{\mathcal{F}_{\mathcal{M}}(T)}$ in terms of $n$,
$\Delta$, and $D$, and discuss the quality of these bounds. As stated
above, it suffices to consider $E_{fix}(T)$ and there is no need to
distinguish the minority and the majority model. The following lemma
is crucial for our results. It allows to recursively compute
$\abs{E_{fix}(T)}$. For a node $v$ define
\[E_{fix}(T,v)=\{F\in E_{fix}(T)\suchthat 2(deg_{F}(v)+1)\le deg(v)\}.\]

%First, we show how to compute $\abs{E_{fix}(T)}$ recursively.

% Let $T^1$ (resp.\ $T^2$) be the tree obtained from $T_1$ by attaching
% a new node $w$ to $v_2$ (resp.\ $v_1$). Then
% $Fix(T_i,v_i)\subseteq E_{fix}(T^i)$ for $i=0,1$ \todo{Brauchen wir
%   das?}

\begin{lemma}\label{lem:basic_red}
  Let $T$ be a tree, $e=(v_1,v_2)\in E^2(T)$, and $T_i$ the
  constituents of $T$ for $e$. Then
  $\abs{E_{fix}(T)}= \abs{E_{fix}(T_1,v_1)}\abs{E_{fix}(T_2,v_2)} +
  \abs{E_{fix}(T_1)}\abs{E_{fix}(T_2)}$.
\end{lemma}
\begin{proof}
  Let $A= \{F\in E_{fix}(T)\suchthat e \in F\}$ and
  $B= \{F\in E_{fix}(T)\suchthat e \not\in F\}$. Then
  \[A \cap B=\emptyset \text{ and }
    A\cup B= E_{fix}(T),\] i.e., $\abs{E_{fix}(T)}= \abs{A}+\abs{B}$.
  If $F\in A$ then $F\setminus e \cap T_i \in E_{fix}(T_i,v_i)$ since
  $T_1\cap T_2=\{e\}$. Hence,
  $\abs{A} \le \abs{E_{fix}(T_1,v_1)}\abs{E_{fix}(T_2,v_2)}$. If
  $F\in B$ then $F \cap T_i \in E_{fix}(T_i)$, i.e.,
  $\abs{B} \le \abs{E_{fix}(T_1)}\abs{E_{fix}(T_2)}$. This yields,
  \[\abs{E_{fix}(T)}\le \abs{E_{fix}(T_1,v_1)}\abs{E_{fix}(T_2,v_2)} +
    \abs{E_{fix}(T_1)}\abs{E_{fix}(T_2)}.\] If
  $F_i \in E_{fix}(T_i,v_i)$, then $F_1\cup F_2 \cup \{e\} \in A$. If
  $F_i \in E_{fix}(T_i)$, then $F_1\cup F_2\in B$. Hence,
  \[\abs{E_{fix}(T_1,v_1)}\abs{E_{fix}(T_2,v_2)} +
  \abs{E_{fix}(T_1)}\abs{E_{fix}(T_2)}\le \abs{E_{fix}(T)}.\]
\end{proof}

% This proves
% the first statement. If $F\in E_{fix}(T_i,v_i)$ then
% $F\cup\{e\} \in E_{fix}(T^i)$.

If $deg_T(v_i)\equiv 0 (2)$, then
$E_{fix}(T_i,v_i) = E_{fix}(T_i\!\setminus\!\{v_i\})$. This yields a
corollary.

% $\abs{E_{fix}(T)}=
% \abs{E_{fix}(T_1\setminus\{v_1\})}\abs{E_{fix}(T_2\setminus\{v_2\})} +
% \abs{E_{fix}(T_1)}\abs{E_{fix}(T_2)}$.

\begin{corollary}\label{cor:path}
  Let $P_n$ be a path with $n$ nodes, then
  $\abs{E_{fix}(P_n)}= F_{n-1}$.
\end{corollary}

\subsection{Computing $\abs{\mathcal{F}_{\mathcal{M}}(T)}$}
Algorithm~1 of \cite{Turau_2022} enumerates all elements of
$E_{fix}(T)$ for a tree $T$. Since $\abs{E_{fix}(T)}$ can grow
exponentially with the size of $T$, it is unsuitable to efficiently
determine $\abs{E_{fix}(T)}$. In this section we propose an efficient
novel algorithm to compute $\abs{E_{fix}(T)}$ in time $O(n\Delta)$
based on Lemma~\ref{lem:basic_red}. The algorithm operates in several
steps. Let us define the input for the algorithm. First, each node
$v_i$ is annotated with $b_i=\lfloor deg(v_i)/2\rfloor$. Let $T_R$ be
the tree obtained form $T$ by removing all leaves of $T$; denote by
$t$ the number of nodes of $T_R$. Select a node of $T_R$ as a root and
assign numbers $1,\ldots,t$ to nodes in $T_R$ using a postorder
depth-first search. Direct all edges towards higher numbers, i.e., the
numbers of all predecessors of a node $i$ are smaller than $i$, see
Fig.~\ref{fig:compFix} for an example. The annotated rooted tree $T_R$
is the input to Algorithm~\ref{fig:algo}.

\begin{figure}[h]
  \begin{center}
 \includegraphics[scale=0.9]{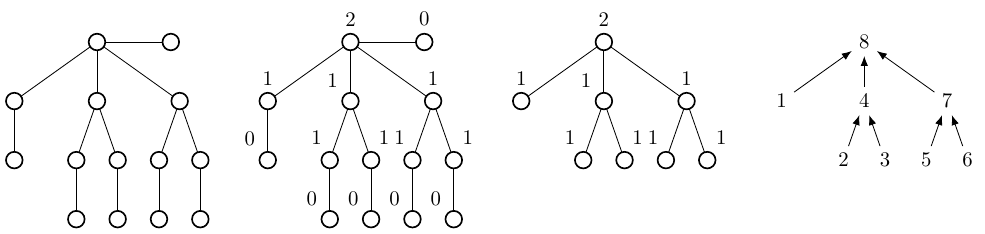}
      \end{center}
      \caption{From left to right: Tree $T$, annotation of $T$,
        $T_R$, a postorder numbering of $T_R$.}\label{fig:compFix}
\end{figure}

Algorithm~\ref{fig:algo} recursively operates on two types of subtrees
of $T_R$ which are defined next. For $k=1,\ldots,t-1$ denote by $T_k$
the subtree of $T_R$ consisting of $k$'s parent together with all
nodes connected to $k$'s parent by paths using only nodes with numbers
at most $k$ (see Fig.~\ref{fig:subtrees}). Note that $T_R=T_{t-1}$.
For $k=1,\ldots,t$ denote by $S_k$ the subtree of $T_R$ consisting of
all nodes from which node $k$ can be reached. In particular
$S_{t}=T_R$, and if $k$ is a leaf then $S_{k}$ consist of node $k$
only. 

% For an inner node $k$ let
% $l(k)=\max \{i\suchthat i \text{ child of } k\}$, if $k$ is a leaf let
% $l(k)=0$. If $k$ is not a leaf then $T_{l(k)}$ is the subtree of $T_R$
% consisting of all nodes from which node $k$ can be reached, in
% particular $T_{l(t)}=T_R$. If $k$ is a leaf then $T_{l(k)}$ consist of
% node $k$ only.

% For $i=1,\ldots,t$ denote by $S_i$ the subtree of $T_R$ consisting of
% all nodes from which node $i$ can be reached, e.g., $S_t=T_R$. Note
% that $i$ is the root of $S_i$. If $i$ is a leaf then $S_i$ consists of
% node $i$ only. Also, if $i$ is not a leaf then $T_{l(i)}=S_i$.

\begin{figure}[h]
  \begin{center}
 \includegraphics[scale=0.95]{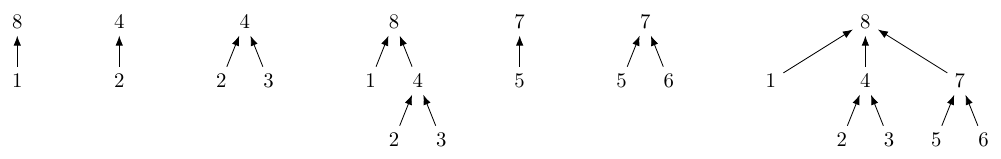}
      \end{center}
      \caption{The subtrees $T_1,\ldots T_7$ of the tree from
        Fig.~\ref{fig:compFix}. Note that $S_4=T_3$ and
        $S_3$ consists of node $3$ only.}\label{fig:subtrees}
\end{figure}
For a subtree $S$ of $T_R$ with largest node $s$ and
$b\ge 0$ denote by $w(S,b)$ the number of subsets $F$ of $E(S)$ with
$deg_{F}(i)\le b_i$ for all nodes $i$ of $S$ (recall that $b_i$ is
defined above) and $deg_{F}(s)\le b$. Let $w(S,-1)=0$. Clearly if
$b\ge b_{s}$ then $w(S,b)=w(S,b_{s})$. If $S$ consists of a single
node $s$ then $s$ is a leaf and $E(S)=\emptyset$; therefore $w(S,b)=1$
for all $b\ge 0$. Note that $w(S,b_s-1) = \abs{E_{fix}(S,s)}$. The
following observation shows the relation between $\abs{E_{fix}(T)}$
and $w(T_{k},b)$.

\begin{lemma}\label{lem:observation}
  $\abs{E_{fix}(T)}=w(T_{t-1},b_t)$ for any tree $T$.
\end{lemma}
The next lemma shows how to recursively compute $w(T_k,b)$ using
Lemma~\ref{lem:basic_red}. 
\begin{lemma}\label{lem:countFix}
  Let $i$ be an inner node of $T_R$, $k$ a child of $i$, and $b\ge 0$.
  Let $\delta_{b,0}=0$ if $b=0$ and $1$ otherwise. If $k$ is the
  smallest child of $i$, then
  \[w(T_k,b)=w(S_{k},b_k) + w(S_{k},b_k-1)\delta_{b,0}.\]
  Otherwise let
  $j\not=k$ be the largest child of $i$ such that $j<k$. Then
  \[w(T_k,b)=w(T_j,b)w(S_{k},b_k) + w(T_j,b-1)w(S_{k},b_k-1).\]
\end{lemma}
%  otherwise \[w(T_k,b)= 2 - \delta_{b,0}\]
\begin{proof}
  The proof for both cases is by induction on $k$. Consider the first
  case. If $k$ is a leaf then $w(S_{k},b)=1$ for all $b\ge 0$ and
  $w(S_{k},-1)=0$. This is the base case. Assume $k$ is not a leaf.
  $T_k$ consists of node $i$, $S_{k}$ and the edge $e=(i,k)$. If $b=0$
  then $\delta_{b,0}=0$ and $w(T_k,0)=w(S_{k},b_k)$ by definition. Let
  $b>0$, i.e., $\delta_{b,0}=1$. Let $f_i$ (resp.\ $f_o$) be the
  number of $F\subseteq E(T_k)$ with $deg_{F}(l)\le b_l$ for all nodes
  $l$ of $T_k$, $deg_{F}(i)\le b$ and $e\in F$ (resp.\ $e\not\in F$).
  Let $F \subseteq E(T_k)$. If $e\in F$ let $\hat{F} = F\setminus e$.
  Then $\hat{F}\in E(S_{k})$, $deg_{\hat{F}}(l)\le b_l$ for all nodes
  $l\not= k$ in $S_{k}$ and $deg_{\hat{F}}(k)\le b_k-1$, hence
  $f_i\le w(S_{k},b_k-1)$. If $e\not\in F$ then $F\in E(S_{k})$ and
  $deg_{F}(l)\le b_l$ for all nodes $l$ in $S_{k}$, hence
  $f_o\le w(S_{k},b_k)$. Thus,
  \[w(T_k,b)\le w(S_{k},b_k) + w(S_{k},b_k-1).\] On the other hand, let
  $F\in E(S_{k})$ such that $deg_{F}(l)\le b_l$ for all nodes
  $l\not= k$ in $S_{k}$. If $deg_{F}(k)\le b_k-1$ then $F\cup \{e\}$
  contributes to $w(T_k,b)$ and if $deg_{F}(k)\le b_k$ then $F$
  contributes to $w(T_k,b)$. Thus
  \[w(T_k,b)\ge w(S_{k},b_k) + w(S_{k},b_k-1).\]

  Consider the second statement. The case that $k$ is a leaf follows
  immediately from Lemma~\ref{lem:basic_red}. Assume that $k$ is not a
  leaf. We apply Lemma~\ref{lem:basic_red} to $T_k$ and edge $(i,k)$.
  Then $T_j \cup (i,k)$ and $S_{k}\cup (i,k)$ are the constituents
  of $T_k$. Note that $E_{fix}(S_{k} \cup (i,k),i) = w(S_{k},b_k-1)$
  and $E_{fix}(T_{j} \cup (i,k),k) = w(T_{j},b-1)$.%\todo{Noch ungenau!}
\end{proof}

% Note that $S_k=T_j$ if $j$ is the largest child of $k$. Thus, for
% increasing values of $k$ and all values of $b\le B$ we compute
% $w(k,b)$ with the equations of Lemma~\ref{lem:countFix}. We have
% $\abs{E_{fix}(T)} =w(t.l,b_t)$ by Lemma~\ref{lem:observation}.

% To simplify notation we define $w(T_{\hat{k}},b) = w(S_{k},b)$ if $k$
% is an inner node and $\hat{k}$ the largest child of $k$. If $k$ is a
% leaf then $w(T_{\hat{k}},b)=1$.

\begin{algorithm}[h]
  %\SetAlgoNoLine
  \DontPrintSemicolon
  \For{$b=0,\dots,B$}{
    $W(0,b)=1$\;
    }
  \For{$k=1,\dots,t-1$}{
    \uIf{$k$ is the smallest child of its parent in $T_R$}{
      $W(k,0):=W(l(k),b_k)$\;
        \For{$b=1,\dots,B$}{
          $W(k,b):=W(l(k),b_k) + W(l(k),b_k-1)$\;
          }
      }
      \Else{
        let $j$ be the largest sibling of $k$ with $j<k$ in $T_R$\;
        $W(k,0):=W(j,0)W(l(k),b_k)$\;
    \For{$b=1,\dots,B$}{
      $W(k,b):=W(j,b)W(l(k),b_k) + W(j,b-1)W(l(k),b_k-1)$\;
      }
    }
  }
  \caption{Computation of $W(k,b)$ for all $k$ and $b$ using
    $T_R$. \label{fig:algo}}
\end{algorithm}

Algorithm~\ref{fig:algo} makes use of Lemma~\ref{lem:observation} and
\ref{lem:countFix} to determine $w(T_{t-1},b_t)$, which is equal to
$\abs{E_{fix}(T)}$. Let $B=\max \{b_i\suchthat i=1,\ldots, t\}$,
clearly $B\le \Delta/2$. Algorithm~\ref{fig:algo} uses an array $W$ of
size $[0,t-1] \times [0,B]$ to store the values of $w(\cdot,\cdot)$.
The first index is used to identify the tree $T_k$. To simplify
notation this index can also have the value $0$. To store the values
of $w(S_{k},b)$ in the same array we define for each inner node $k$ an
index $l(k)$ as follows $l(k)=k-1$ if $k$ is not a leaf and $l(k)=0$
otherwise. Then clearly $S_k=T_{l(k)}$ if $k$ is not a leaf. More
importantly, the value of $w(S_{k},b)$ is stored in $W({l(k)},b)$ for
all $k$ and $b$.

The algorithm computes the values of $W$ for increasing values of
$k< t$ beginning with $k=0$. If $W(j,b)$ is known for all $j<k$ and
all $b\in [0,B]$ we can simply compute $W(k,b)$ for all values of $b$
in $[0,B]$ using the equations of Lemma~\ref{lem:countFix}. Finally we
have $w(t-1,b_t)$ which is equal to $\abs{E_{fix}(T)}$. See
Appendix~\ref{app:exec1} for an execution of Algorithm~\ref{fig:algo}.
Theorem~\ref{theo:1} follows from Lemma~\ref{lem:observation} and
\ref{lem:countFix}.

\subsection{Upper Bounds for $\abs{\mathcal{F}_{\mathcal{M}}(T)}$}
The definition of $E_{fix}(T)$ immediately leads to a first upper
bound for $\abs{\mathcal{F}_{\mathcal{M}}(T)}$.

\begin{lemma}\label{lem:firstBoundF}
  $\abs{E_{fix}(T)} \le 2^{n-\Delta-1}$.
\end{lemma}
\begin{proof}
  Theorem 1 of \cite{Turau_2022} implies
  $\abs{E_{fix}(T)} \le 2^{\abs{E^2(T)}}$. Note that
    $\abs{E^2(T)}=n-1-l$, where $l$ is the number of leaves of
    $T$. It is well known that $l=2+\sum_{j=3}^\Delta (j-2)D_j$, where
    $D_j$ denotes the number of nodes with degree $j$. Thus,
  \[\hspace*{7mm}\abs{E^2(T)} = n-3-\sum_{j=3}^\Delta (j-2)D_j\le n-3 -(\Delta-2) =
    n-\Delta-1.\hspace*{15mm}\]
\end{proof}
% The last lemma implies $\abs{\mathcal{F}_{\mathcal{M}}(T)}\le 2^i$ for
% $\Delta(T)=n-i$.

For $\Delta < n -\lceil n/3\rceil$ the bound of
Lemma~\ref{lem:firstBoundF} is not attained. Consider the case
$n=8, \Delta=4$. The tree from Fig.~\ref{fig:exectFix2} with $x=1$ has
$14$ fixed points, this is the maximal attainable value.

\begin{figure}[h]
  \begin{center}
 \includegraphics[scale=0.95]{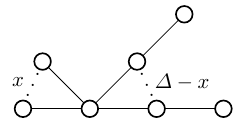}    
  \end{center}
  \caption{If $x\ge \Delta/2$ then
    $\abs{\mathcal{F}_{\mathcal{M}}(T)}=
    2^{n-\Delta}$.}\label{fig:exectFix2}
\end{figure}

For $\Delta < n/2$ we will prove a much better bound than that of
Lemma~\ref{lem:firstBoundF}. For this we need the following technical
result.
\begin{lemma}\label{lem:star}
  Let $T=(V,E)$ be a tree with a single node $v$ that has degree
  larger than $2$. Let ${\cal D}=[dist(w,v)\suchthat w\in V, w\not=v]$
  be the multi-set with the distances of all nodes to $v$. Then
  \[\abs{E_{fix}(T)}= \sum_{S\subset {\cal D},  \abs{S}\le
      \Delta/2}~\prod_{s\in S} \mathbb{F}_{s-1} \prod_{s\in {\cal D}\setminus
      S}\mathbb{F}_{s}.\]
  \end{lemma}
  \begin{proof}
    Let ${\cal P}$ be the set of all $\Delta$ paths from $v$ to a
    leaf of $T$. Let $F\in E_{fix}(T)$ and
    \[{\cal P}_F= \{P \in {\cal P} \suchthat F \cap P \not\in
    E_{fix}(P)\}.\] Let $P\in {\cal P}_F$ and $v_P$ the node of $P$
    adjacent to $v$. Then $(v,v_P)\in F$. This yields
    $\abs{{\cal P}_F}\le \Delta/2$. For $P\in {\cal P}$ let $\hat{P}$
    be an extension of $P$ by an edge $(v,x)$ with a new node $x$.
    Then $F \cap P \in E_{fix}(\hat{P})$ for each $P\in {\cal P}_F$.
    By Cor.~\ref{cor:path} there are $\mathbb{F}_{\abs{P}-2}$ possibilities for
    $F\cap P$. Let $\bar{{\cal P}}_F= {\cal P} \setminus {\cal P}_F$.
    For $P\in \bar{{\cal P}}_F$ we have $(v,v_P)\not \in F$. Hence,
    $F \cap P \in E_{fix}(P)$. By Cor.~\ref{cor:path} there are
    $\mathbb{F}_{\abs{P}-1}$ possibilities for $F\cap P$. Also
    $\abs{\bar{{\cal P}}_F}= \Delta - \abs{{\cal P}_F}$.

    Let ${\cal P}_1 \subset {\cal P}$ with
    $\abs{{\cal P}_1}\le \Delta/2$ and $\hat{F}_P\in E_{fix}(\hat{P})$
    for all $P\in {\cal P}_1$ and $F_P\in E_{fix}(P)$ for all
    $P\in {\cal P}\setminus {\cal P}_1$. Then the union of all
    $\hat{F}_P$ and all $F_P$ is a member of $E_{fix}(T)$. This yields
    the result.
  \end{proof}

\begin{corollary}\label{cor:star}
  Let $T$ be a $l$-generalized star graph. Then
  \[\abs{E_{fix}(T)}= \sum_{i=0}^{\lfloor \Delta/2 \rfloor}\binom{\Delta}{i}
    \mathbb{F}_{l-1}^i \mathbb{F}_{l}^{\Delta-i}.\] 
\end{corollary}

The corollary yields that a star graph (i.e.\ $l=1$) has two fixed
points. For $l=2$ we have the following result.

\begin{lemma}\label{lem:genStar}
  Let $T$ be a $2$-generalized star graph. Then
  $\abs{E_{fix}(T)}\le \mathbb{F}_{n-\lceil\Delta/2\rceil}$.
\end{lemma}
\begin{proof}
  We use Corollary~\ref{cor:star}. If
$\Delta\equiv 0(2)$ then
\[\abs{E_{fix}(T)}= \sum_{i=0}^{\lfloor \Delta/2
    \rfloor}\binom{\Delta}{i} = \frac{1}{2}\left(2^\Delta + \binom{\Delta}{\Delta/2}\right)
  \le \mathbb{F}_{3\Delta/2+1} = \mathbb{F}_{n-\Delta/2},\] otherwise
$\abs{E_{fix}(T)} =2^{\Delta-1} \le
\mathbb{F}_{n-\lceil\Delta/2\rceil}$.
\end{proof}

% Let $T$ be a $2$-generalized star graph. If $\Delta\equiv 0(2)$ then
% \[\abs{E_{fix}(T)}= \sum_{i=0}^{\lfloor \Delta/2
%     \rfloor}\binom{\Delta}{i} = \frac{1}{2}\left(2^\Delta + \binom{\Delta}{\Delta/2}\right)
%   \le F_{3\Delta/2+1} = F_{n-\Delta/2},\] otherwise
% $\abs{E_{fix}(T)} =2^{\Delta-1} \le
% F_{n-\lceil\Delta/2\rceil}$.

In Theorem~\ref{theo:2} we prove that the upper bound of
Lemma~\ref{lem:genStar} holds for all trees. First, we prove two
technical results.

\begin{lemma}\label{lem:rem}
  Let $T$ be a tree and $v$ a leaf of $T$ with neighbor $w$. Let $n_l$
  (resp.\ $n_{i}$) be the number of neighbors of $w$ that are leaves
  (resp.\ inner nodes). If $n_l>n_{i}$ then
  $E_{fix}(T)=E_{fix}(T\setminus v)$.
\end{lemma}
\begin{proof}
  Clearly, $E_{fix}(T\setminus v)\subseteq E_{fix}(T)$. Let
  $F\in E_{fix}(T)$. Then $deg_{F}(w)\le n_i$. Thus,
  $2deg_{F}(w)\le 2n_i\le n_l-1+n_i=deg_T(w)-1$. Hence,
  $F\in E_{fix}(T\setminus v)$, i.e.,
  $E_{fix}(T)\subseteq E_{fix}(T\setminus v)$.
\end{proof}

\begin{lemma}\label{lem:3path}
  Let $T$ be a tree and $v,w,u$ a path with $deg(v)=1$ and $deg(w)=2$.
  Then $\abs{E_{fix}(T)}\le\abs{E_{fix}(T_v)}+\abs{E_{fix}(T_w)}$ with
  $T_v=T\setminus v$ and $T_w=T_v\setminus w$.
\end{lemma}
\begin{proof}
  Let $F\in E_{fix}(T)$ and $e=(u,w)$. If $e\in F$ then
  $F\setminus e \in E_{fix}(T_w)$ otherwise $F \in E_{fix}(T_v)$. This
  proves the lemma.
\end{proof}

\begin{lemma}\label{lem:count_fix}
  $\abs{E_{fix}(T)} \le \mathbb{F}_{n-\lceil\Delta/2\rceil}$ for a tree $T$
  with $n$ nodes.
\end{lemma}
\begin{proof}
  The proof 
  is by induction on $n$. If $\Delta=2$ the result holds by
  By Cor.~\ref{cor:path}. If $T$ is a star graph then
  $\abs{\mathcal{F}_{\mathcal{{M}}}(T)}=2$, again the result is true.
  Let $\Delta> 2$ and $T$ not a star graph. Thus, $n>4$. There exists
  an edge $(v,w)$ of $T$ where $v$ is a leaf and all neighbors of $w$
  but one are leaves. If $deg(w)>2$ then there exists a neighbor
  $u\not=v$ of $w$ that is a leaf. Let $T_u=T\setminus u$. Then
  $\abs{E_{fix}(T)}=\abs{E_{fix}(T_u)}$
  by Lemma~\ref{lem:rem}. Since $\Delta(T_u)\ge\Delta(T)-1$ we have by
  induction
  \[\abs{E_{fix}(T)}=\abs{E_{fix}(T_u)}
  \le \mathbb{F}_{n-1-\lceil\Delta(T_u)/2\rceil}\le
  \mathbb{F}_{n-\lceil\Delta(T)/2\rceil}.\] Hence, we can assume that
  $deg(w)=2$.

  Let $u\not=v$ be the second neighbor of $w$. Denote by $T_v$ (resp.\
  $T_w$) the tree $T\setminus v$ (resp.\ $T\setminus \{v,w\}$). By
  Lemma~\ref{lem:3path} we have
  \[\abs{E_{fix}(T)}\le\abs{E_{fix}(T_v)}+\abs{E_{fix}(T_w)}.\] If there
  exists a node different from $u$ with degree $\Delta$ then
  \[\abs{E_{fix}(T)}\le \mathbb{F}_{n-1-\lceil\Delta/2\rceil} +
  \mathbb{F}_{n-2-\lceil\Delta/2\rceil} = \mathbb{F}_{n-\lceil\Delta/2\rceil}\] by
  induction. Hence we can assume that $u$ is the only node with degree
  $\Delta$. Repeating the above argument shows that $T$ is
  $2$-generalized star graph with center node $u$. Hence,
  $\abs{E_{fix}(T)} \le \mathbb{F}_{n-\lceil\Delta/2\rceil}$ by
  Lemma~\ref{lem:genStar}.
\end{proof}

Lemmas~\ref{lem:firstBoundF} and \ref{lem:count_fix} prove the upper
bound of Theorem~\ref{theo:2}. The bound of Lemma~\ref{lem:count_fix}
is sharp for paths. Note that for a fixed value of $n$ the monotone
functions $\mathbb{F}_{n-\lceil\Delta/2\rceil}$ and $2^{n-\Delta}$ intersect in
$\Delta \in (\lceil n/2\rceil, \lceil n/2\rceil+1)$. 

\subsection{Lower Bounds for
  $\abs{\mathcal{F}_{\mathcal{M}}(T)}$}\label{sec:lowerFix}
A trivial lower bound for $\abs{\mathcal{F}_{\mathcal{M}}(T)}$ for all trees is
$1+\abs{E_{fix}(T)}$. It is sharp for star graphs. For a better bound
other graph parameters besides $n$ are required.

\begin{lemma}\label{lem:lower1}
  Let $T$ be a tree and $T_L$ the tree obtained from $T$ by removing
  all leaves. Then $\abs{E_{fix}(T)}\ge 2^{r/2}$,
  where $r$ is the number of inner nodes of $T_L$.
\end{lemma}
\begin{proof}
  By induction we prove that $T_L$ has a matching $M$ with $r/2$
  edges. Then $M\subseteq E_{fix}(T)$ and each subset of $M$ is {\em
    $\mathcal{F}$-legal}.
\end{proof}

Applying Lemma~\ref{lem:lower1} to a 2-generalized star graph yields a
lower bound of $1$, which is far from the real value. Another lower
bound for $\abs{E_{fix}(T)}$ uses $D$, the diameter of $T$. Any tree
$T$ with diameter $D$ contains a path of length $D+1$. Thus,
$\abs{E_{fix}(T)}\ge F_{D}$ by Cor.~\ref{cor:path}. This completes the
proof of Theorem~\ref{theo:2}. We show that there are trees for which
$\abs{E_{fix}(T)}$ is much larger than $F_{D}$. Let
$n,D,h\in \mathbb{N}$ with $n-2 > D \ge 2(n-1)/3$ and
$n-D-1\le h\le 2D-n+1$. Let $T_{n,D,h}$ be a tree with $n$ nodes that
consists of a path $v_0,\ldots,v_D$ and another path of length $n-D-2$
attached to $v_h$. Clearly, $T_{n,D,h}$ has diameter $D$. Also
$deg(v_h)=3$, all other nodes have degree $1$ or $2$.
Fig.~\ref{fig:hdTree} shows $T_{n,D,h}$.

\begin{figure}[h]
  \begin{center}
  \includegraphics[scale=0.9]{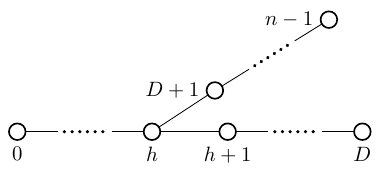}
\end{center}
  \caption{A tree $T_{n,D,h}$.}\label{fig:hdTree}
\end{figure}

\begin{lemma}\label{lem:nDh}
  $\abs{E_{fix}(T_{n,D,h})} = \mathbb{F}_D\mathbb{F}_{n-D-1} + \mathbb{F}_{h}
  \mathbb{F}_{D-h}\mathbb{F}_{n-D-2}$.\end{lemma}
\begin{proof}
  Let $e=(v_h,v_{D+1})$ and $T_{v_h},T_{v_{D+1}}$ the constituents of
  $T$ for $e$. Then
  \[\abs{E_{fix}(T)}= \abs{E_{fix}(T_{v_h},v_h)}\abs{E_{fix}(T_{v_{D+1}},v_{D+1})}
    + \abs{E_{fix}(T_{v_h})}\abs{E_{fix}(T_{v_{D+1}})}\] by
  Lemma~\ref{lem:basic_red}. Clearly,
  $\abs{E_{fix}(T_{v_h},v_h)}= \mathbb{F}_h\mathbb{F}_{D-h}$ and
  $\abs{E_{fix}(T_{v_{D+1}},v_{D+1})}= \mathbb{F}_{n-D-2}$ by
  Cor.~\ref{cor:path}.
\end{proof}

The next lemma determines values of $h$ for which
$\abs{E_{fix}(T_{n,D,h})}$ is maximal.
\begin{lemma}\label{lem:tnD}
  Let $n,D\in \mathbb{N}$ with $n-2>D \ge 2(n-1)/3$. Let
  \[h_0=\argmax\limits_{n-D-1\le h\le
      2D-n+1}\abs{E_{fix}(T_{n,D,h})}.\] Then $h_0=n-D-1$ if
  $3D = 2(n-1)$ and otherwise $h_0$ is the first odd integer of
  $[n-D-1, 2D-n+1]$. Furthermore,
  $\abs{{E}_{fix}(T_{n,D,h_0})} = \mathbb{F}_D\mathbb{F}_{n-D-1} + \mathbb{F}_{h_0}
  \mathbb{F}_{D-h_0}\mathbb{F}_{n-D-2}$.
\end{lemma}
\begin{proof}
  By Lemma~\ref{lem:nDh} we have
  $\abs{\mathcal{F}_{\mathcal{{M}}}(T_{n,D,h})} = \mathbb{F}_D\mathbb{F}_{n-D-1} + \mathbb{F}_{h}
  \mathbb{F}_{D-h}\mathbb{F}_{n-D-2}$. Thus, only the term $\mathbb{F}_{h} \mathbb{F}_{D-h}$ depends on
  $h$. By Lemma~\ref{lem:fib1} $\mathbb{F}_h \mathbb{F}_{D-h}$ it is maximal if $h=h_0$.
\end{proof}
\begin{lemma}\label{lem:fib1}
  Let $0<x < c$. Then $\mathbb{F}_x\mathbb{F}_{c-x}>\mathbb{F}_{x+1}\mathbb{F}_{c-x-1}$ if $x$ is odd and
  $\mathbb{F}_x\mathbb{F}_{c-x}<\mathbb{F}_{x+1}\mathbb{F}_{c-x-1}$ otherwise.
\end{lemma}
\begin{proof}%[Proof of Lemma~\ref{lem:fib1}]
  The proof is by induction on $x$. Let $x=1$. Then
  \[\mathbb{F}_1\mathbb{F}_{c-1}=\mathbb{F}_{c-1} >\mathbb{F}_{c-2}=\mathbb{F}_{2}\mathbb{F}_{c-2}.   \]
  Let $x>1$. Then
  \begin{align*}
\mathbb{F}_x\mathbb{F}_{c-x}&=\mathbb{F}_{x+1}\mathbb{F}_{c-x} - \mathbb{F}_{x-1}\mathbb{F}_{c-x}\\ &= \mathbb{F}_{x+1}\mathbb{F}_{c-x-1}
  +\mathbb{F}_{x+1}\mathbb{F}_{c-x-2} - \mathbb{F}_{x-1}\mathbb{F}_{c-x} \\ &= \mathbb{F}_{x+1}\mathbb{F}_{c-x-1}
  +\mathbb{F}_{x}\mathbb{F}_{c-x-2} +\mathbb{F}_{x-1}\mathbb{F}_{c-x-2} - \mathbb{F}_{x-1}\mathbb{F}_{c-x} \\ &=
  \mathbb{F}_{x+1}\mathbb{F}_{c-x-1} +\mathbb{F}_{x}\mathbb{F}_{c-x-2} - \mathbb{F}_{x-1}\mathbb{F}_{c-x -1}.    
  \end{align*}
   Now we can
  apply induction. If $x$ is odd, then
  $\mathbb{F}_{x-1}\mathbb{F}_{c-x-1}<\mathbb{F}_{x}\mathbb{F}_{c-x-2}$ and hence,
  $\mathbb{F}_{x}\mathbb{F}_{c-x} >\mathbb{F}_{x+1}\mathbb{F}_{c-x-1} $. Is $x$ is even, then
  $\mathbb{F}_{x-1}\mathbb{F}_{c-x-1}>\mathbb{F}_{x}\mathbb{F}_{c-x-2}$ and hence,
  $\mathbb{F}_{x}\mathbb{F}_{c-x} <\mathbb{F}_{x+1}\mathbb{F}_{c-x-1} $.
\end{proof}

Let $T_{n,D}$ be a tree that maximizes the number of fixed points
among all trees with $n$ nodes and diameter $D$. An interesting
question is about the structure of $T_{n,D}$. By an exhaustive search
among all trees with $n\le 34$ using a computer there was just a
single case where this was not a star-like tree (all nodes but one
have degree $1$ or $2$) that maximizes the number of fixed points (see
Fig.~\ref{fig:n32_d7}). We have the following conjecture.

\begin{conjecture}
  Except for a finite number of cases for each combination of $n$
  and $D$ there exists a star-like graph that maximizes the number of
  fixed points.
\end{conjecture}

\begin{figure}[h]
  \begin{center}
  \includegraphics[scale=0.6]{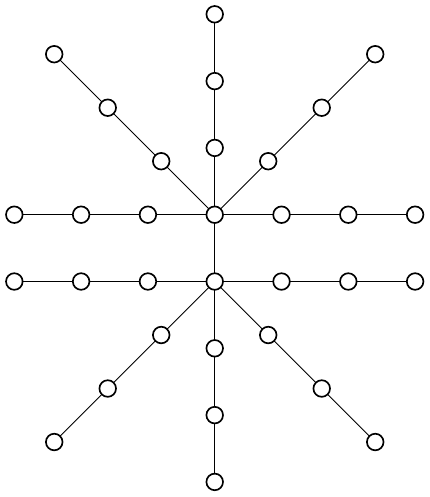}
\end{center}
  \caption{A tree with 32 nodes and diameter 7 with 181376 fixed
    points. All other trees with 32 nodes and diameter 7 have less
    fixed points.}\label{fig:n32_d7}
\end{figure}

% The $(k,\Delta)$-{\em broom graph} $B_{k,\Delta}$ depicted in
% Fig.~\ref{fig:broom} shows that the bound of
% Theorem~\ref{theo:2} comes for all values of $\Delta$ within a
% factor of two to the real maximal value of
% $\abs{\mathcal{F}_{\mathcal{{M}}}(T)}$. For this graph we have
% \[\abs{\mathcal{F}_{\mathcal{{M}}}(B_{k,\Delta})}=2\left(F_{k-3}\sum_{i=0}^{\lfloor
%       \Delta/2\rfloor} \binom{\Delta-1}{i} +
%     F_{k-2}\sum_{i=0}^{\lfloor \Delta/2\rfloor}
%     \binom{\Delta}{i}\right).\] Using
% $2\sum_{i=0}^{\lfloor \Delta/2\rfloor} \binom{\Delta}{i} = 2^\Delta +
% \binom{\Delta}{\Delta/2}$ it is easy to show that
% $4F_{n-\lceil\Delta/2\rceil} \ge
% \abs{\mathcal{F}_{\mathcal{{M}}}(B_{k,\Delta})}$

% \begin{figure}[h]
%   \includegraphics[scale=0.88]{broom}
%   \caption{The $(k,\Delta)$-broom graph.}\label{fig:broom}
%  \end{figure}

\subsection{Special Cases}
%  If
% $x < \Delta/2$, i.e., $\Delta \ge n/2$, the number of
% fixed points of the trees from Fig.~\ref{fig:exectFix} is\todo{Beweisen!}
% \[2\sum_{i=0}^{\min \left(\Delta/2, n-\Delta-1\right)}
%   \binom{n-\Delta-1}{i}.\] This is the maximal attainable value for
% this combination of $n$ and $\Delta$. The tree in
% Fig.~\ref{fig:Exception} shows that for $n=2\Delta +1$ there are trees
% with more fixed points then the trees in Fig.~\ref{fig:exectFix}.

% \begin{figure}[h]
%  \includegraphics[scale=0.8]{automat}
%  \caption{This tree has 36 fixed point, this is maximal for $n=11$ and
%    $\Delta=5$ ZEICHNEN.}\label{fig:exectFix}
% \end{figure}

\begin{lemma}\label{lem:special1}
  Let $T$ be a tree with $n\ge 4$ and $\Delta \ge n -\lceil n/3\rceil$
  then $\abs{E_{fix}(T)}\le 2^{n-\Delta-1}$. This bound is sharp.
\end{lemma}
\begin{proof}
  First we construct a tree $T_m$ realizing this bound. Let $T_m$ be a
  tree with a single node $v$ with degree $\Delta$, $x=2\Delta-n+1$
  neighbors of $v$ are leaves and the remaining $\Delta-x=n-\Delta-1$
  neighbors have degree $2$ (see Fig.~\ref{fig:exectFix2}). Assumption
  $\Delta \ge n -\lceil n/3\rceil$ implies $\Delta -x \le \Delta/2$.
  Hence
  $E_{fix}(T_m)=
  \sum_{i=0}^{\Delta-x}\binom{\Delta-x}{i}=2^{\Delta-x}=
  2^{n-\Delta-1}$ (see also Lemma~\ref{lem:star}).

  Let $v$ be a node of $T$ with degree $\Delta$. Assume that at most
  two neighbors of $v$ are leaves. Then
  $n\ge \Delta + 1 + \Delta -2= 2\Delta -1$. This yields
  $\Delta \le (n+1)/2$, which contradicts the assumption
  $\Delta \ge n -\lceil n/3\rceil$. Hence, at least three neighbors of
  $v$ are leaves. Let $w$ be a non-leaf neighbor of $v$ and $e=(v,w)$.
  Without loss of generality we can assume that $deg(w)=2$ and that
  the neighbor $v'\not= v$ is a leaf. Next we apply
  Lemma~\ref{lem:basic_red}. Let $x$ be a neighbor of $v$ that is a
  leaf. Note that $E_{fix}(T_v,u)= E_{fix}(T_v\setminus \{w,x\})$. By
  induction we have $\abs{E_{fix}(T_v,u)}\le 2^{n-\Delta-2}$. We also
  have $\abs{E_{fix}(T_v)}\le 2^{n-\Delta-2}$. This yields the upper
  bound.
\end{proof}

\begin{lemma}\label{lem:fix_sharp}
  Let $T$ be a tree with
  $n -\lceil n/3\rceil > \Delta > (n-1)/2$ then
  $\abs{E_{fix}(T)} \le \sum_{i=0}^{\lfloor\Delta/2\rfloor}
  \binom{n-\Delta-1}{i}$. This bound is sharp.
\end{lemma}
\begin{proof}
  $n -\lceil n/3\rceil > \Delta$ yields $\Delta/2 < n/3$ and
  $\Delta +n/3 +1 \le n$. Therefore
  $\lfloor\Delta/2\rfloor < n -\Delta -1$.
  Let $U$ be a tree satisfying the assumption and that has the maximal
  number of fixed points among all such trees. Let $c$ be a node of
  degree $\Delta$. Assume there exist a neighbor $u$ of $c$ such that
  the subtree $T_u$ rooted in $u$ contains at least three nodes. Since
  $2\Delta +1 > n$, there  exists a neighbor $x$ of $c$ that is a
  leaf. Let $v$ be a leaf of $T_u$ and $v_s$ its single neighbor such that all neighbors of $v_s$ but
  one are leaves. Let $w$ be this non-leaf neighbor. Assume that
  $deg(w)>2$. Then there exits a leaf $v' \in N(w)$ with $v'\not=v$.
  Detaching node $v'$ from $w$ and attaching it to $x$ would increase
  the number of fixed points. This contradiction shows that
  $deg(w)=2$.

  Assume that $w\not = u$. Le $w_1 \not= v$ a neighbor of $w$. Using
  the same argument as above we can assume that $deg(w_1)=2$. Let
  $e=(w,w_1)$. Consider the case with and without $e$. We can apply
  Lemma~\ref{lem:basic_red}. This yields
  \[\abs{E_{fix}(T)} = \abs{E_{fix}(T\setminus \{v\})}
  +\abs{E_{fix}(T\setminus \{v,w\})}.\] Note that both of these trees
  have maximal degree $\Delta(T)$. Each of these trees either
  satisfies the assumption of Lemma~\ref{lem:fix_sharp} or
  Lemma~\ref{lem:special1}. In the former case we can apply induction
  to get a bound for the number of fixed points. Combining these bound
  yields the result.

  It remains to consider the case $w=u$. Then each neighbor is either
  a leaf or has a single neighbor that also is a leaf. In this case
  Lemma~\ref{lem:star} yields the result. This type of tree also shows
  that the bound is sharp.
\end{proof}

Let $\tau_{n,\Delta}$ be the maximal value of $\abs{E_{fix}(T)}$ for
all trees $T$ with $n$ nodes and maximal degree $\Delta$. We have the
following conjecture. If this conjecture is true, it would be possible
to determine the structure of all trees with
$\abs{E_{fix}(T)}=\tau_{n,\Delta}$.

\begin{conjecture}
  $\tau_{n,\Delta} = \tau_{n-1,\Delta} +\tau_{n-2,\Delta}$ for
  $\Delta < (n-1)/2$.
\end{conjecture}

% Assume that the conjecture is true. For given $\Delta$ choose smallest
% $n$ such that $\Delta< (n-1)/2$. If $n$ is odd then
% $\Delta =(n-1)/2 -1=(n-3)/2$. Thus, $n= 2\Delta+3$. Hence,
% Lemma~\ref{lem:fix_sharp} allows to compute $\tau_{n-2,\Delta}$. If
% $n$ is even then $n=2\Delta$ and this is also true. Thus, we only need
% to compute $\tau_{2(\Delta+1),\Delta}$ for all $\Delta$ (plus the above
% exceptions). \todo{Das stimmt so noch nicht!}

 \section{Pure 2-Cycles}
 In this section we prove an upper bound for
 $\abs{\mathcal{P}_{\mathcal{M}}(T)}$. Again we use the fact that
 $\abs{\mathcal{P}_{\mathcal{M}}(T)} = 2\abs{E_{pure}(T)}$
 \cite{Turau_2022}. Note that Algorithm~\ref{fig:algo} can be easily
 adopted to compute $\abs{E_{pure}(T)}$. The difference between {\em
   $\mathcal{F}$-legal} and {\em $\mathcal{P}$-legal} is that instead
 of $2deg_{F}(v) \le deg(v)$ condition $2deg_{F}(v) < deg(v)$ is required. If
 $deg(v)$ is odd then the conditions are equivalent. Thus, it suffices
 to define for each node $v_i$ with even degree
 $b_i=\lfloor deg(v_i)/2\rfloor-1$. Hence, $\abs{E_{pure}(T)}$ can be
 computed in time $O(n\Delta)$.

 Note that $E_{pure}(T) \subseteq E_{fix}(T)$ with
 $E_{pure}(T)=E_{fix}(T)$ if all degrees of $T$ are odd. Thus,
 $\abs{E_{pure}(T)}\le F_{n-\lceil\Delta/2\rceil}$. In this section we
 prove the much better general upper bound stated in
 Theorem~\ref{theo:3}. We start with an example. Let $n\equiv 0(2)$
 and $H_n$ the tree with $n$ nodes consisting of a path $P_n$ of
 length $n/2$ and a single node attached to each inner node of $P_n$
 (see Fig.~\ref{fig:kamm}). Since all non-leaves have degree $3$ we
 have $E_{pure}(H_n)=E_{fix}(P_n)$, thus,
 $\abs{E_{pure}(H_n)}\le F_{n/2}$ by Corollary~\ref{cor:path}. We
 prove that $F_{\lfloor n/2\rfloor}$ is an upper bound for
 $\abs{E_{pure}(T)}$ in general.

\begin{figure}[h]
  \hfill%
  \includegraphics[scale=0.95]{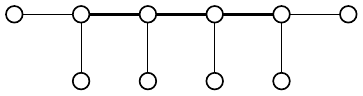}
  \hfill\null%
  \caption{The graph $H_{10}$, edges of $E^3(H_{10})$ are depicted
    as solid lines.}\label{fig:kamm}
\end{figure}

We first state a few technical lemmas and then prove
Theorem~\ref{theo:3}. The proof of the first Lemma is similar to
Lemma~\ref{lem:rem}.
\begin{lemma}\label{lem:remPure} %[Proof in Appendix]
  Let $T$ be a tree, $v$ a leaf with neighbor $w$, and
  $T'=T\setminus v$. Let $n_l$ (resp.\ $n_{i}$) be the number of
  neighbors of $w$ that are leaves (resp.\ inner nodes). If
  $n_l>n_{i}+1$ then $\abs{E_{pure}(T)}=\abs{E_{pure}(T')}$.
\end{lemma}

% \begin{lemma}[Proof in Appendix]
%   \label{lem:Fib_rel}
%   For each $k\ge 0$ and $c\ge 2k+3$ we have
%   $F_{c-(2k+3)}=\sum_{i=0}^{k}(-1)^i \binom{k}{i} F_{c-(i+3)}$.
% \end{lemma}

% \begin{lemma}[Proof in Appendix]
%   \label{lem:Fib_sum}
%   For each $c\ge 1$ and $d \le \frac{c+1}{2}$ we have
%   $\sum_{i=1}^{d}F_{c-(2i-1)}\le F_c$.exist
% \end{lemma}

For a tree $T$ let $T^2$ be the tree obtained from $T$ by recursively
removing each node $v$ with degree $2$ and connecting the two
neighbors of $v$ by a new edge. Note that $T^2$ is uniquely defined
and $deg_{T^2}(v) = deg_T(v)$ for each node $v$ of $T^2$. The next
lemma is easy to prove.
\begin{lemma}\label{lem:pureT2}
$E_{pure}(T)\subseteq E_{pure}(T^2)$ for each tree  $T$.
\end{lemma}
\begin{proof}
  Let $w$ be a node $v$ of $T$ with degree $2$. Let $T'$ be the tree
  obtained from $T$ by removing $w$ and connecting the two neighbors
  of $w$ by a new edge. Let $F \in E_{pure}(T)$. Since $deg(w)=2$ no
  edge of $F$ is incident to $w$. Thus, for all $v\not=w$ we have
  $2deg_{F}(v)< deg_T(v)=deg_{T'}(v)$, i.e., $F\in E_{pure}(T')$. Hence,
  $E_{pure}(T)\subseteq E_{pure}(T')$ and the statement follows by
  induction.
\end{proof}

\begin{proof}[Proof of  Theorem~\ref{theo:3}]
  Proof by induction on $n$. The statement is true for $n\le 5$ as can
  be seen by a simple inspection of all cases. Let $n\ge 6$. By
  Lemma~\ref{lem:pureT2} we can assume that no node of $T$ has degree
  $2$. Let $\hat{T}$ be the tree induced by the edges in $E^3(T)$.
  $\hat{T}$ includes all inner nodes of $T$. Let $u$ be a node of
  $\hat{T}$ such that all neighbors of $u$ in $\hat{T}$ except one are
  leaves. Denote the neighbors of $u$ in $\hat{T}$ that are leaves by
  $v_1,\ldots,v_d$ with $d\ge 1$, i.e., $deg_{\hat{T}}(u)=d+1$. Let
  $e_i=(u,v_i)$. By Lemma~\ref{lem:remPure} we can assume
  $deg_T(v_i)=3$ for $i=1,\ldots,d$. Denote the two neighbors of $v_i$
  in $T\setminus \hat{T}$ by $v_i^a$ and $v_i^b$. Let
  $H_1=\{F\in E_{pure}(T) \suchthat e_1\not\in F\}$ and
  $T'=T\setminus \{v_1^a,v_2^b\}$. Clearly $H_1=E_{pure}(T')$. Thus,
  $\abs{H_1}\le \mathbb{F}_{\lfloor n/2\rfloor-1}$ by induction. Assume that
  $u$ has a neighbor $u'$ in $T$ that is a leaf in $T$. Let
  $F\in E_{pure}(T)\setminus H_1$. Then
  $F\setminus \{e_1\}\in E_{pure}(T\setminus \{u',v_1,v_1^a,v_1^b\})$.
  Thus, by induction
  $\abs{E_{pure}(T)\setminus H_1} \le \mathbb{F}_{\lfloor n/2\rfloor-2}$ and
  hence, $\abs{E_{pure}(T)}\le \mathbb{F}_{\lfloor n/2\rfloor}$. Therefore we
  can assume that $deg_T(u) = d+1$, i.e., $d\ge 2$.

  Next we expand the definition of $H_1$ as follows. For
  $i=1,\ldots,d$ let
  \[H_i= \{F\in E_{pure}(T)\suchthat e_i\not\in F \text{ and }
  e_1,\ldots, e_{i-1} \in F\}.\] Thus, $H_i=\emptyset$ for
  $i\ge (d+3)/2$ since $2deg_F(u) < d+1$. Let
  $d_0= \lfloor (d+1)/2\rfloor$. Then
  $E_{pure}(T)=H_1 \cup \ldots \cup H_{d_o}$. We claim that
  $\abs{H_i}\le \mathbb{F}_{\lfloor n/2\rfloor-(2i-1)}$. The case $i=1$ was
  already proved above. Let $i\ge 2$ and define
  $T_i= T\setminus \{v_j,v_j^a,v_j^b\suchthat j=i-1,i\}$. Then
  $\abs{E_{pure}(T_i)}\le \mathbb{F}_{\lfloor n/2\rfloor-3}$ by induction,
  hence $\abs{H_2}=\abs{E_{pure}(T_2)}\le \mathbb{F}_{\lfloor n/2\rfloor-3}$.
  Let $i\ge 3$ and $F\in H_i$. Then
  $F\setminus \{e_{i-1}\} \in E_{pure}(T_i)$. Clearly,
  $E_{pure}(T_i)\setminus H_i$ consists of all $F \in  E_{pure}(T_i)$
  for which $\{e_1,\ldots, e_{i-2}\} \cap F \not= \emptyset$ holds.
  Hence, 
  \[\abs{H_i} \le \abs{E_{pure}(T_i)}- \abs{\bigcup_{j=1}^{i-2}
      E_j^i}\le \mathbb{F}_{\lfloor n/2\rfloor-3} - \abs{\bigcup_{j=1}^{i-2}
      E_j^i}\] where
  $E_j^i=\{F\in E_{pure}(T_i) \suchthat e_{j}\not\in F\}$. We use the
  following well known identity for $I=\{1,\ldots,i-2\}$.
\[\abs{\bigcup_{i\in I} E_i}=\sum_{\emptyset \not= J \subseteq I}
  (-1)^{\abs{J}+1} \abs{\bigcap_{j\in J} E_i}\] Note that
$\abs{\bigcap_{j\in J} E_j^i}\le \mathbb{F}_{\lfloor n/2\rfloor-(\abs{J}+3)}$ by
induction. Let $k=i-2$. Then
\[\abs{H_i} \le \mathbb{F}_{\lfloor n/2\rfloor-3} - \sum_{j=1}^{k}(-1)^{j+1}
  \binom{k}{j} \mathbb{F}_{\lfloor n/2\rfloor-(j+3)}= \sum_{j=0}^{k}(-1)^{j}
  \binom{k}{j} \mathbb{F}_{\lfloor n/2\rfloor-(j+3)}\] Note that
$\lfloor n/2\rfloor \ge 2k+3$. Lemma~\ref{lem:Fib_rel} 
implies $\abs{H_i} \le \mathbb{F}_{\lfloor n/2\rfloor-(2i-1)}$. Then
Lemma~\ref{lem:Fib_sum}  yields
  \[\abs{E_{pure}(T)} \le \sum_{i=1}^{d_0} \abs{H_i}\le
    \sum_{i=1}^{d_0} \mathbb{F}_{\lfloor
  n/2\rfloor-(2i-1)} \le \mathbb{F}_{\lfloor n/2\rfloor}.\]
\end{proof}

\begin{lemma}
  \label{lem:Fib_rel}
  For each $k\ge 0$ and $c\ge 2k+3$ we have
  \[\mathbb{F}_{c-(2k+3)}=\sum_{i=0}^{k}(-1)^i \binom{k}{i} \mathbb{F}_{c-(i+3)}.\]
\end{lemma}
\begin{proof}%[Proof of Lemma~\ref{lem:Fib_rel}]
  The proof is by induction on $k$. The cases $k\le 2$ clearly hold.
  Let $k>2$.
  \begin{align*}
  &\sum_{i=0}^{k}(-1)^i \binom{k}{i} \mathbb{F}_{c-(i+3)} \\&~= (-1)^k\mathbb{F}_{c-(k+3)} +
    \sum_{i=0}^{k-1}(-1)^i \binom{k}{i} \mathbb{F}_{c-(i+3)} \\
  &~= (-1)^k\mathbb{F}_{c-(k+3)} +
    \sum_{i=0}^{k-1}(-1)^i \binom{k-1}{i} \mathbb{F}_{c-(i+3)} +
    \sum_{i=1}^{k-1}(-1)^i \binom{k-1}{i-1} \mathbb{F}_{c-(i+3)} \\
  &\txteq{Ind.}(-1)^k\mathbb{F}_{c-(k+3)} +\mathbb{F}_{c-(2k+1)} - \sum_{i=0}^{k-2}(-1)^i
    \binom{k-1}{i} \mathbb{F}_{c-(i+4)}\\
  &~=(-1)^k\mathbb{F}_{c-(k+3)} +\mathbb{F}_{c-(2k+1)} - \sum_{i=0}^{k-1}(-1)^i
    \binom{k-1}{i} \mathbb{F}_{c-1-(i+3)} + (-1)^{k-1}\mathbb{F}_{c-(k+3)}\\
    &\txteq{Ind.}\mathbb{F}_{c-(2k+1)} - \mathbb{F}_{c-1-(2k+1)}\\
    &~= \mathbb{F}_{c-(2k+3)}
\end{align*}
\end{proof}

\begin{lemma}
  \label{lem:Fib_sum}
  For each $c\ge 1$ and $d \le \frac{c+1}{2}$ we have
  $\sum_{i=1}^{d}\mathbb{F}_{c-(2i-1)}\le \mathbb{F}_c$.
\end{lemma}
\begin{proof}%[Proof of Lemma~\ref{lem:Fib_sum}]
  The proof is by induction on $d$. The cases $d\le 2$ clearly hold.
  Let $d>2$. By induction we get
  \[\sum_{i=1}^{d}\mathbb{F}_{c-(2i-1)}= \mathbb{F}_{c-1} + \sum_{i=2}^{d}\mathbb{F}_{c-(2i-1)} = \mathbb{F}_{c-1} + \sum_{i=1}^{d-1}\mathbb{F}_{c-2-(2i-1)}
    \txtle{Ind.} \mathbb{F}_{c-1} + \mathbb{F}_{c-2} = \mathbb{F}_c.\]
\end{proof}

%\todo{Conjecture: Maximal number of pure 2-cycles is always star-like
% graph, false tree\_32\_d7}

\subsection{Special Cases}\label{sec:spca}

\begin{lemma}%[Proof in Appendix]
  Let $T$ a tree with $n$ nodes and diameter $D$ such that $2D\ge n$.
  Then $\abs{E_{pure}(T)} \le \mathbb{F}_{n-D}$ and this bound is sharp.\label{lem:pureDiam}
\end{lemma}
\begin{proof}%[Proof of Lemma~\ref{lem:pureDiam}]
  Let $P$ be a path of length $D$ of $T$. Thus, at most $n-(D+1)$
  nodes can have degree $3$. Each such node can be associated with a
  unique node of degree $1$. Thus, at most $n-(D+1)$ neighboring nodes
  on $P$ have degree 3. This leads to $n-D$ adjacent edges from
  $E^3(T)$ and hence, $\abs{E_{pure}(T)} \le \mathbb{F}_{n-D}$. The trees $H_n$
  (see Fig.~\ref{fig:kamm}) prove that the bound is sharp.
\end{proof}

If $n> 2D$ this bound does not hold. It is easy to see that for
$n=2D+1$ the maximal number of fixed points is $\mathbb{F}_{n-(D+1)}$ (see
Fig.~\ref{fig:simpleP}).

\begin{figure}[h]
  \hfill%
  \includegraphics[scale=0.95]{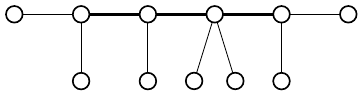}
  \hfill\null%
  \caption{The graph $H_{10}$ with an additional node has 10 pure
    2-cycles.}\label{fig:simpleP}
\end{figure}

%\todo{Geht auch für $n$ odd und $n \le 2\Delta +1$, ein leaf mehr!}
\begin{lemma}
  Let $n < 2\Delta +1$ and $T$ a tree with $n$ nodes and maximal
  degree $\Delta$. Then
  $\abs{E_{pure}(T)} \le 2^{\lfloor\frac{n-\Delta-1}{2}\rfloor}$. This
  bound is sharp.
\end{lemma}
\begin{proof} Let $x \le \Delta$ such that $2(\Delta-x) < \Delta$. Let
  $v$ be a node of degree $\Delta$ that has $x$ leaves as neighbors
  and the remaining $\Delta-x$ neighbors have degree $3$, i.e.,
  $x=(1+3\Delta-n)/2$. Then clearly
  $\abs{E_{pure}(T)} = 2^{\Delta -x}$. Note that
  $\Delta -x = (n-\Delta-1)/2$.
\end{proof}

%Similar conjecture as above

%\subsection{Another Bound}
For a tree $T$ let $\delta_T$ be the smallest degree of an inner node
of $T$ with degree at least $3$, i.e.,
$\delta_T=\min \{deg(v)\suchthat v\in V \text{ with } deg(v)>2\}$.
\begin{lemma}%[Proof in Appendix]
  Let $T$ be a tree with $\Delta>2$ and $n'$ the number of nodes of $T^2$. Then
  $\abs{E_{pure}(T)} \le 2^{(n'-2)/(\delta_T-1)-1}$.\label{lem:pureB3}
\end{lemma}
\begin{proof}%[Proof of Lemma~\ref{lem:pureB3}]
  The proof is by induction on $t$, the number of nodes of $T$ with
  degree $2$. Let $t=0$, then $n=n'$. Denote the number of leaves
  (resp.\ inner nodes) by $l$ (resp.\ $i$). Then clearly
  $l \ge (\delta_T - 2)i+2$ and $i\le (n-2)/(\delta_T-1)$. This yields
  \[\abs{E^3(T)}=n-1-l= i-1 \le (n-2)/(\delta_T-1) -1.\] Therefore,
  $\abs{E_{pure}(T)} \le 2^{(n-2)/(\delta_T-1) -1}$. Assume $t>0$ and
  let $w$ be a node of $T$ with $deg(w)=2$. Let $T'$ be the tree
  obtained from $T$ by removing $w$ and connecting the two neighbors
  of $w$ by a new edge. Then $(T')^2=T^2$ and $\delta_{T'}=\delta_T$.
  By induction $\abs{E_{pure}(T')} \le 2^{(n'-2)/(\delta_T-1)-1}$.
  Since $\abs{E_{pure}(T)} \le \abs{E_{pure}(T')}$ (see
  Lemma~\ref{lem:pureT2}) the proof is complete.
\end{proof}

For $\delta_T \ge 4$ we have
$\mathbb{F}_{n/2} > 2^{(n-2)/(\delta_T-1) -1}$, i.e., a better bound
than that provided in Theorem~\ref{theo:3}. If $\delta_T=3$ the other
bound is lower since $\mathbb{F}_{n/2} \le 2^{n/2-2}$.

% In Appendix~\ref{sec:spca} we list more bounds for $E_{pure}(T)$ that
% depend on specific graph parameters. We also provide examples showing
% that the bounds are sharp.

\section{Conclusion and Open Problems}

The problem of counting the fixed points for general cellular automata
is \#P-complete. In this paper we considered counting problems
associated with the minority/majority rule of tree cellular automata.
In particular we examined fixed points and pure 2-cycles. The first
contribution is a novel algorithm that counts the fixed points and the
pure 2-cycles of such an automata. The algorithms run in time
$O(\Delta n)$. It utilizes a characterization of colorings for these
automata in terms of subsets of the tree edges. This relieved us from
separately treating the minority and the majority rule. The second
contribution are upper and lower bounds for the number of fixed points
and pure 2-cycles based on different graph parameters. We also
provided examples to demonstrate the cases when these bounds are
sharp. The bounds show that the number of fixed points (resp.\
2-cycles) is a tiny fraction of all colorings.

There are several open questions that are worth pursuing. Firstly, we
believe that it is possible to sharpen the provided bounds and to
construct examples for these bounds. In particular for the case
$\Delta \le n/2$, we believe that our bounds can be improved. Another
line of research is to extend the analysis to generalized classes of
minority/majority rules. One option is to consider the
minority/majority not just in the immediate neighborhood of a node but
in the r-hop neighborhood for $r>1$ \cite{moran1994r}. We believe that
this complicates the counting problems considerably.

Finally the predecessor existence problem and the corresponding
counting problem \cite{barrett2007} have not been considered for tree
cellular automata. The challenge is to find for a given
$c \in \mathcal{C}(T)$ a coloring $c' \in \mathcal{C}(T)$ with
$\mathcal{M}(c')=c$. A coloring $c \in \mathcal{C}(T)$ is called a
{\em garden of Eden} coloring if there doesn't exist a
$c' \in \mathcal{C}(T)$ with $\mathcal{M}(c')=c$. The corresponding
counting problem for tree cellular automata is yet unsolved.

 \bibliographystyle{splncs04}
 \bibliography{minority}

%\newpage
\appendix
\section{Execution of Algorithm~\ref{fig:algo} }%{Appendix}
\label{app:exec1}

In this appendix we give an example for the execution of
Algorithm~\ref{fig:algo}. We use the tree of Fig.~\ref{fig:compFix}
which is for convenience shown again in Fig.~\ref{fig:compFix_App}.

\begin{figure}[h]
  \begin{center}
 \includegraphics[scale=0.9]{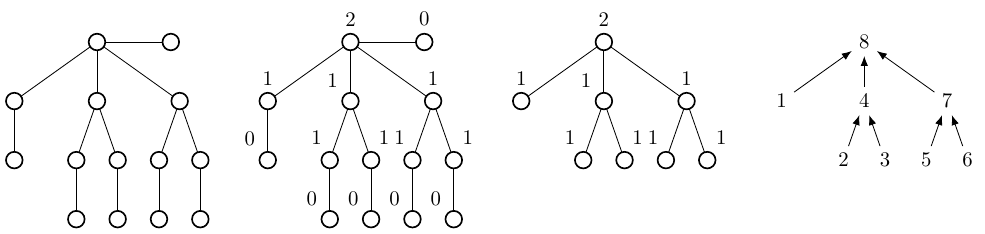}
      \end{center}
      \caption{From left to right: Tree $T$, annotation of $T$,
        $T_R$, a postorder numbering of $T_R$.}\label{fig:compFix_App}
\end{figure}

For the tree in Fig.~\ref{fig:compFix_App} we have
\begin{align*}
  w(T_7,2)&=w(T_4,2)w(S_{7},1) + w(T_4,1)w(S_{7},0)=w(T_4,2)w(T_{6},1) + w(T_4,1)w(T_{6},0)\\
  w(T_4,2)&=w(T_1,2)w(S_{4},1) + w(T_1,1)w(S_{4},0)=w(T_1,2)w(T_{3},1) + w(T_1,1)w(T_{3},0)\\
  w(T_1,2)&=w(S_1,1) + w(S_1,0) =2\\
  w(T_3,1)&=w(T_2,1)w(S_3,1) + w(T_2,0)w(S_3,0) = w(T_2,1) + w(T_2,0)\\
  w(T_2,1)&=w(S_2,1)+w(S_2,0)= 1+1=2\\
  w(T_2,0)&=w(S_2,1)= 1\\
  w(T_3,1)&=3\\
  w(T_1,1)&=w(S_1,1) + w(S_1,0) =2\\
  w(T_3,0)&=w(T_2,0)w(S_3,1) + w(T_2,-1)w(S_3,0) = 1\\
  w(T_4,1)&=w(T_1,1)w(S_4,1)+ w(T_1,0)w(S_4,0)=w(T_1,1)w(T_3,1)+ w(T_1,0)w(T_3,0)=2*3+1*1=7
\end{align*}
Hence, $w(T_4,2)=8$, $w(T_7,2)=31$, and therefore
$\abs{\mathcal{F}_{\mathcal{M}}(T)}=62$. The complete array $W$ is as
follows.
\renewcommand{\arraystretch}{1.5}
\setlength{\tabcolsep}{5pt}
\begin{center}
\begin{tabular}[h]{c||c|c|c|c|c|c|c|c}
{ $T_k$ }& $0$ & $1$ & $2$ & $3$ & $4$ & $5$ & $6$ & $7$ \\
\hline \hline
 $b=0$ &$1$ &  1 &1 &1 &3 &1 &1 &9  \\
\hline
 $b=1$ &$1$ &  2 &2 &3 &7 &2 &3 &24  \\
\hline
 $b=2$ &$1$ &  2 &2 &3 &8 &2 &3 &31  \\
\end{tabular}
\end{center}

\end{document}